\documentclass[11pt, oneside,margin=1in]{article}   	
\usepackage{geometry}                		
\usepackage{graphicx}				
\usepackage{amssymb}
\usepackage{amsmath,amsthm}
\usepackage[colorlinks=true]{hyperref}
\usepackage{enumitem}

\usepackage{xcolor}
\usepackage{comment}
\newcommand{\na}[1]{{\color{blue} #1}}
\newcommand{\mattnote}[1]{{\color{red} #1}}
\newtheorem{theorem}{Theorem}

\newtheorem{observation}{Observation}
\newtheorem{remark}{Remark}
\newtheorem{lemma}{Lemma}
\newtheorem{example}{Example}
\newtheorem{corollary}{Corollary}

\newtheorem{proposition}{Proposition}
\title{\vspace{-1 in} Bitcoin: A Natural Oligopoly}

\hypersetup{%
   citecolor=black
}

\author{Nick Arnosti\thanks{Columbia University. Email: nicholas.arnosti@gmail.com.} \and Matt Weinberg\thanks{Princeton University. Email: smweinberg@princeton.edu. Supported by NSF CCF-1717899.}}

\begin{document}
\maketitle
\vspace{-.2 in}

\begin{abstract}

Although Bitcoin was intended to be a decentralized digital currency, in practice, mining power is quite concentrated. This fact is a persistent source of concern for the Bitcoin community.

We provide an explanation using a simple model to capture miners' incentives to invest in equipment. In our model, $n$ miners compete for a prize of fixed size. Each miner chooses an investment $q_i$, incurring cost $c_i q_i$, and then receives reward $\frac{q_i^\alpha}{\sum_j q_j^\alpha}$, for some $\alpha \geq 1$. When $c_i = c_j$ for all $i,j$, and $\alpha = 1$, there is a unique equilibrium where all miners invest equally. However, we prove that under seemingly mild deviations from this model, equilibrium outcomes become drastically more centralized. In particular, 
\begin{itemize}
\item When costs are asymmetric, if miner $i$ chooses to invest, then miner $j$ has market share at least $1-\frac{c_j}{c_i}$. That is, if miner $j$ has costs that are (e.g.) $20\%$ lower than those of miner $i$, then miner $j$ must control at least $20\%$ of the \emph{total} mining power.
\item In the presence of economies of scale ($\alpha > 1$), every market participant has a market share of at least $1-\frac{1}{\alpha}$, implying that the market features at most $\frac{\alpha}{\alpha - 1}$ miners in total.
\end{itemize}

We discuss the implications of our results for the future design of cryptocurrencies. In particular, our work further motivates the study of protocols that minimize ``orphaned'' blocks, proof-of-stake protocols, and incentive compatible protocols. 
\end{abstract}
\maketitle
\addtocounter{page}{-1}
\newpage
\section{Introduction}\label{sec:intro}
In the ten years since Bitcoin was introduced by \cite{nakamoto_2008}, cryptocurrencies have become increasingly mainstream, accounting for hundreds of billions of dollars in market capitalization\footnote{Source: \url{www.coinmarketcap.com}.} and increasing amounts of media coverage. The protocol underlying Bitcoin is undoubtedly clever, and has largely succeeded in creating a public yet anonymized record of transactions. 
However, Bitcoin has largely failed to deliver on one key promise: that of true decentralization. While there there is nothing in Bitcoin's source code that enforces centralized control (as opposed to e.g. Visa, which is centralized by design), Bitcoin's record of transactions is in fact dictated by a small number of participants. More specifically, Bitcoin relies on entities called ``miners" to record and verify transactions. Today, a majority of transactions are verified by one of four large ``mining pools."\footnote{Source: \url{https://blockchain.info/pools?timespan=4days}.} 
For reasons discussed below, this concentration of power has caused great concern among the Bitcoin community, and prompted some to argue forcefully for new cryptocurrencies. For example, one recent whitepaper wrote\footnote{Source: \url{https://bravenewcoin.com/assets/Whitepapers/NxtWhitepaper-v122-rev4.pdf}.}
\begin{quote} \it Bitcoin's creator, Satoshi Nakamoto, intended for the bitcoin network to be fully decentralized, but nobody could have predicted that the incentives provided by Proof of Work systems would result in the centralization of the mining process. \end{quote}

One reason that centralization of mining power might be considered unexpected is that mining rewards exhibit decreasing marginal gains: \emph{the more invested one currently is in Bitcoin mining, the less one has to gain from additional investment.} This is because each new mining device reduces the rewards earned by all existing devices. Our paper shows that this force is generally dwarfed by two other factors: {\em asymmetric costs} and {\em economies of scale}.

In the context of Bitcoin, mining costs are closely related to the cost of electricity, which varies substantially with geography. Most bitcoin mining is based in China, where some of the cheapest electricity in the world can be found near underutilized hydroelectric plants.\footnote{Source: \url{https://www.coinbureau.com/analysis/much-bitcoin-mining-concentrated-china/}.} Economies of scale arise naturally due to non-linearity in the cost of storing, powering, and cooling hardware, and also potentially due to strategic behavior~\cite{EyalS14, KiayiasKKT16, SapirshteinSZ16, BabaioffDOZ12, CarlstenKNW16, Eyal15}. Our main results demonstrate that even seemingly small cost imbalances or super-linear scaling of reward with investment result in a highly concentrated market for bitcoin mining. In particular, 





\begin{itemize}
\item Corollary~\ref{cor:main} states that if one agent can acquire computational power at, say, $20\%$ lower cost than another, then if the less efficient agent chooses a non-zero investment, the more efficient agent must control \emph{at least 20\% of all computational power}. 
\item Theorem~\ref{thm:bound} states that when rewards are proportional to each agent's investment raised to some power $\alpha > 1$, then in every equilibrium, at most $1+\frac{1}{\alpha-1}$ agents choose to invest at all. 
\end{itemize} 




Below, we provide a brief overview of the relevant aspects of Bitcoin before proceeding to describe our model and main results.

\subsection{What is Bitcoin Mining?}\label{sec:mining}

Bitcoin was designed to replace centralized ``digital currencies'' like Paypal or Visa. While many people (including the authors) are perfectly happy to use these systems, others are repulsed at the thought of their transaction fees going to a corporate CEO, or are concerned that the US government might pressure an organization to freeze a user's funds or reveal their transaction history. The promise of Bitcoin was that no single entity would be in control of the ledger, but instead millions of miners in a peer-to-peer network would each contribute to maintaining a record of past transactions, also known as a ``ledger."

There are two fundamental challenges facing Bitcoin and other cryptocurrencies. First, how to ensure consistency among the many copies of the transaction history? Second, how to incentivize users to store and update this information? Achieving consensus is a notoriously difficult problem studied by computer scientists for several decades~\cite{FischerLP85,LamportSP82,PeaseSL80}, and is especially challenging in \emph{permissionless environments} (where users can anonymously join and leave as they please). 

In order to ensure consistency, the Bitcoin protocol makes it difficult to amend the ledger. In particular, each ``block" of new transactions must be accompanied by the solution to a cryptopuzzle. The puzzle is designed in such a way that the fastest way to find a solution is simply to guess at random,\footnote{At least this is widely believed to be the case, and is true under the assumption that SHA-256 is an ideal hash function~\cite{BitcoinBook}.} so agents are selected to add to the ledger \emph{in proportion to their computational power}. 

In order to compensate agents for expenses associated with solving cryptopuzzles, the Bitcoin protocol gives a ``block reward" (currently 12.5 bitcoin) to agents who solve a puzzle. Users may also include an optional transaction fee that is also paid to the miner whose block authorizes their transaction. The process of randomly guessing solutions to a cryptopuzzle in search of rewards is referred to as \emph{mining}, and agents who participate in this process are called \emph{miners} -- terminology that we adopt in the remainder of the paper. Importantly, the Bitcoin protocol has a hard-coded ``difficulty adjustment'' for the cryptopuzzles so that one is solved every ten minutes in expectation, regardless of how much total computational power is in the network. 



While the Bitcoin protocol has indeed succeeded in attracting miners to maintain consistent copies of the transaction ledger, it is questionable whether this ledger can be considered ``decentralized." At the time of writing, four mining pools account for more than half of all blocks, and six mining pools account for more than three-quarters of all blocks. 
While the hardware behind a mining pool is not owned by a single entity, a recent study concluded that eleven ``large mining organizations" control over half of global mining power \cite{HilemanR17}. These miners could in principle freeze any user's funds, or wipe past transactions from the record (for more information on how this could happen, see~\cite{BitcoinBook}).






\subsection{Overview of Results}

The discussion above highlights several key features of the Bitcoin protocol, which form the basis for our model: (1) The total value of rewards available to miners is fixed.\footnote{Technically, payout to miners is determined by both the block reward and by transaction fees. For most of Bitcoin's history (including the present day), transaction fees have been trivially small compared to the block reward (although this was not the case for a brief period in late 2017 and early 2018 -- see \url{https://bitinfocharts.com/comparison/bitcoin-transactionfees.html}). Additionally, \cite{Huberman-Leshno-Moallemi_2017} observe that the fees paid in equilibrium should depend only on the throughput of the network, the number of users who want to send transactions, and their delay costs, and {\em not} on miners' investment decisions. Therefore, it's reasonable to treat the prize shared by miners as exogenous to their behavior. We discuss this point further in the conclusion.} (2) Miners make costly investments in computational power. (3) Miners are rewarded (roughly) in proportion to their computational power. The simplest model capturing these features is as follows:
\begin{itemize}
\item There is a fixed reward of value $1$ to be split among $n$ miners.
\item Each miner $i$ decides on computational power to acquire, $q_i$, incurring cost $cq_i$. 
\item Miner $i$ receives reward $\frac{q_i}{\sum_j q_j}$. 
\end{itemize}

In this model, there is a unique equilibrium where each miner invests the same amount, and possesses a $1/n$ fraction of total mining power (see Corollary \ref{cor:symmetric}). This decentralized outcome arises because investment in mining power exhibits diminishing marginal returns.  Additional investment lowers the value of earlier investments --- in the extreme case, if one is already responsible for 100\% of the computational power, there is \emph{no} marginal gain from increased mining power.


In practice, of course, costs are not perfectly symmetric nor perfectly proportional to mining hardware. Still, so long as both assumptions are ``approximately" satisfied, one might reasonably  conjecture that mining power should remain distributed among many parties. Our main (qualitative) message is that seemingly minor departures from these assumptions result in surprisingly high levels of centralization. 



In Section \ref{sec:proportional}, we introduce a model with asymmetric costs and rewards proportional to investment, and show that there is a unique equilibrium (Theorem \ref{thm:unique2}). Corollary \ref{cor:main} states that if miners $i$ and $j$ both purchase mining equipment in this equilibrium, and miner $i$ has per-unit costs that are $x\%$ lower than those of miner $j$, then miner $i$ must possess at least $x\%$ of the {\em total} computational power. We further provides a few examples to help the reader appreciate that with even moderate cost asymmetries, this results in most mining power being controlled by only a few miners.

 In Section \ref{sec:disproportional}, we model economies of scale by assuming that rewards are proportional to $q_i^{\alpha}$ for some $\alpha > 1$. This captures the fact that powering and cooling equipment is cheaper on a per-unit basis, as well as the fact that clever deviations from the mining protocol possibly allow large miners to earn more than his or her ``fair share" of rewards.\footnote{One could explicitly model each of these factors by assuming that when miner $i$ acquires computing power $q_i$, he or she pays a cost $c_i q_i^{\beta}$ and earns a reward of $q_i^{\alpha'}/\sum_j q_j^{\alpha'}$, for some $\beta < 1 < \alpha'$. A simple change of variables reveals that this is strategically equivalent to our model with $\alpha = \alpha'/\beta$, so we work with this simpler representation.} In this case, the results of Corollary \ref{cor:main} grow even starker. Absent economies of scale, it was possible for many miners to participate so long as their costs were sufficiently similar. With economies of scale, we show in Corollary \ref{cor:dfoc} that each participant must have a market share of at least $1 - 1/\alpha$, implying that the number of participants in equilibrium is at most $\frac{\alpha}{\alpha - 1}$, even if all miners have identical costs. In other words, even if $\alpha = 100/99$ ---  implying that the second piece of mining hardware costs $98.6\%$ as much as the first -- then in equilibrium, every miner who chooses to participate will earn at least $1\%$ of all rewards.

Sections~\ref{sec:proportional} and~\ref{sec:disproportional} require almost no background knowledge of cryptocurrencies, and provide standalone insight to centralization in markets where agents compete over a fixed prize. In Section~\ref{sec:conclusions}, we discuss the impact of our results for the future design of cryptocurrencies. Our main conclusion is that the currently-observed concentration of mining power is not a temporary aberration, but rather a natural result of the economic incentives established by Bitcoin and other ``proof of work" based protocols.

\section{Related Work}

There are two classes of related works to discuss. The first concerns academic and non-academic writing about Bitcoin, which touches on some of the themes explored here, but bears no mathematical resemblance to our work. The second consists of existing literature on ``rent-seeking contests", which is mathematically closely related to our work but does not discuss cryptocurrencies.

\subsection{Cryptocurrencies}
Numerous sources have empirically documented high levels of concentration in Bitcoin mining. The most frequently-cited statistics come from publicly available data on large mining pools.\footnote{https:blockchain.info/pools?timespan=4days .} At the time of writing, four mining pools are responsible for over 50\% of mined blocks, and six are responsible for over 75\%. While these statistics are staggering, an informed reader might observe that pool managers do not own all of the the hardware mining on their behalf, and argue that these numbers are therefore misleading.

Other measures of centralization are not publicly available, and are often intentionally obscured. Nevertheless, prior work has concluded that 56\% of Bitcoin \emph{nodes}\footnote{A node listens for transactions and blocks, and forwards them to the rest of the network (additionally checking for validity before forwarding). 
} exist in large data centers~\cite{GencerBERS18}, and that approximately $100$ nodes are responsible for the initial announcement of 75\% all blocks. Even more pertinent to our work is a recent report on the distribution of mining power, which identifies eleven ``large mining entities" and estimates that they collectively control a majority of global computational power dedicated to Bitcoin mining~\cite{HilemanR17}.

\subsection{Market Participation}

Our model is formally identical to that proposed by \cite{tullock_1980}. That paper prompted a large literature on ``imperfectly discriminating contests" (so-named to highlight the contrast with the ``perfectly discriminating contest" of an all-pay auction). Most of the follow-up literature focuses on the extent of {\em rent dissipation} -- that is, comparing expenses incurred by contestants to the value of the prize. Two papers \cite{hillman-riley_1989, gradstein_1995} explicitly address the number of entrants, as we do. Proposition 5 in \cite{hillman-riley_1989} is similar to our Theorem \ref{thm:unique2}, although it does not explicitly characterize the market share of each participant. \cite{gradstein_1995} assumes that costs are drawn iid, and shows that as the number of potential entrants $n$ grows, the fraction who choose to enter in equilibrium converges to zero. Later, he also considers a variant of the model in which entrants incur a small fixed cost.

Our contribution relative to \cite{hillman-riley_1989, gradstein_1995} is three-fold. First, we provide an arguably cleaner characterization of equilibrium via Theorem~\ref{thm:unique2}. Second, we study the case of economies of scale ($\alpha > 1$), whereas those papers focus only on the proportional model, and therefore have no analog to Theorem \ref{thm:bound} or Corollary \ref{cor:dfoc}. Third, we introduce these tools to the TCS community via connection to Bitcoin mining, and note that they have implications for the design of future cryptocurrencies.

%
%


\section{Asymmetric Costs and Proportional Rewards}\label{sec:proportional}
In this section, we consider a model with asymmetric costs and rewards proportional to investment, which we refer to as the \emph{proportional model}, and observe that the  symmetry in the toy model of Section~\ref{sec:intro} is crucial for a decentralized equilibrium. Formally, our model is the following: There are $n \geq 2$ miners competing for a prize of value $1$. First, each miner $i$ chooses an investment level $q_i$, paying cost $c_i q_i$ to do so. Then miners earn market share in proportion to $q_i$ (specifically, miner $i$ receives reward $q_i/\sum_j q_j$). The costs $c_i > 0$ are common knowledge. Without loss of generality, assume that $c_1 \leq c_2 \leq  \cdots \leq c_n$, and define $c_{n+1} = \infty$ for ease of notation. The utility for miner $i$ when investments are given by $q$ is
\[U_i(q) = x_i(q) - c_i q_i,\]
where $x_i(q) = q_i/ \sum_j q_j$, and $x_i(q) = 0$ for all $i$ if $q_i = 0$ for all $i$.
The main result of this section characterizes a unique equilibrium outcome. To state this outcome, we define

\begin{equation} X(c) = \sum_i \max(1 - c_i/c,0). \label{eq:x} \end{equation}

\begin{lemma} \label{lem:monotone}
There is a unique value $c^*$ satisfying $X(c^*) = 1$, and $c^* > c_2$.
\end{lemma}

\begin{proof}
Observe that $X$ is continuous and non-decreasing in $c$. Furthermore, $X$ is strictly increasing on $[c_1, \infty)$, with $X(c_2) = 1 - c_1/c_2 < 1$ and $\lim_{c \rightarrow \infty} X(c) = n$.
\end{proof}

\begin{theorem} \label{thm:unique2}
 In the proportional model, there is a unique pure strategy equilibrium. In it, $q_i = \frac{1}{c^*}\max(1 - c_i/c^*,0)$, and $x_i(q) = \max(1-c_i/c^*,0)$.
\end{theorem}
Note that $c^*$ serves as a participation threshold: miners invest in equilibrium if and only if their cost is less than $c^*$. The total amount of mining hardware purchased is $1/c^*$.

Theorem~\ref{thm:unique2} immediately implies the following corollaries, stated in Section~\ref{sec:intro}. 

\begin{corollary} \label{cor:symmetric}
In the proportional model, if miners have identical costs ($c = c_1 = c_2 = \cdots = c_n$), then all $n$ miners participate and invest such that $q_ic_i = \frac{n-1}{n^2}$ for all $i$.
\end{corollary}

\begin{corollary}\label{cor:main}
In the proportional model, if miner $i$ participates at all in the unique equilibrium (that is, $q_i > 0$), then $x_j(q) \geq 1-\frac{c_j}{c_i}$ for all $\ell$. That is, the market share of miner $\ell$ is at least $1-\frac{c_j}{c_i}$. 
\end{corollary}
\begin{proof}
By Theorem~\ref{thm:unique2}, $x_j(q) = 1-c_j/c^*$ and $c_i < c^{*}$ for any miner who participates.
\end{proof}

One way to think about Corollary \ref{cor:main} is as follows. Suppose that we wish to know whether miner $k$ will participate in equilibrium. For each lower-cost miner, ask ``by what percentage are this miner's costs lower than those of miner $k$?" Miner $k$ will participate if and only if the sum of these percentages is less than 100. For example, if there are three miners with costs that are $20\%$ lower than $k$'s, and five more with costs that are $10\%$ lower, then $k$ will not participate, as $3 \times 20 + 5 \times 10 = 110 > 100$. 

Before proving Theorem~\ref{thm:unique2}, we provide a few quick examples illustrating the degree of centralization with various cost structures. Each of these examples is designed so that costs are ``not too different." In the first example, all $n$ miners have costs at most twice that of the most efficient miner, and yet only $7$ miners choose to participate.

\begin{example} Let $c_i = i/(i+1)$. We have $c_7 < c^* \approx 0.88 <  c_8$, so $7$ miners participate. Moreover, $c_7 = 7/8$, so we can immediately conclude from Corollary~\ref{cor:main} that miner $1$ controls more than $1-\frac{1/2}{7/8} = 3/7$ of the market. 
Miner $2$ controls more than $1-\frac{2/3}{7/8} = 5/21$ of the market, so jointly, they control more than $2/3$ of the market. 

\end{example} 

In the following example, all $n$ miners participate, but mining power is highly concentrated among the largest miners.

\begin{example}\label{ex:exp}Let $c_i = 1-2^{-i}$. Note that $X(1) = \sum_i (1 - c_i) < 1$, from which it follows that $c^* > 1$ (because $X(c^*) = 1$ and $X$ is non-decreasing). Therefore, all miners participate in equilibrium. However, in equilibrium we have $x_i = 1 - c_i/c^* \geq 1 - c_i = 2^{-i}$. So in particular, Miner 1 controls at least half the mining power, and together with Miner 2 controls at least $3/4$. 


\end{example}


\subsection{Proof of Theorem~\ref{thm:unique2}}
The proof of Theorem~\ref{thm:unique2} follows mostly from a clever manipulation of first-order conditions. First, we state a helpful lemma showing that first-order conditions are necessary and sufficient for identifying pure-strategy equilibria.

\begin{lemma} \label{lem:foc}
Suppose that $\tilde{q}$ is an equilibrium in the proportional model. Then for each $i$, $\sum_{j \neq i} \tilde{q}_j > 0$, and 
\begin{equation} x_i(\tilde{q}) = \max( 1 - c_i \sum_i \tilde{q}_i,0). \label{eq:xi} \end{equation}
Furthermore, any $\tilde{q}$ satisfying~\eqref{eq:xi} for all $i$ is an equilibrium. 
\end{lemma}

\begin{proof}
We first prove that if $\tilde{q}$ is an equilibrium, then for each $i$ we have $\sum_{j \neq i} \tilde{q}_j > 0$. If this were not the case, then no $\tilde{q}_i >0$ can be a best response, as $i$ could do better by choosing $q_i = \tilde{q}_i/2$. However, $\tilde{q}_i = 0$ also cannot be a best response, as $i$ could do better with any $q_i \in (0,1/c_i)$.

Next, note that $U_i(q) = x_i(q) - c_iq_i$ is continuous and differentiable in $q_i$ on $\{q : \sum_{j \neq i} q_j > 0\}$. Slightly abusing notation, for the remainder of this proof we define
\begin{equation} x_i'(q) = \frac{\partial x_i(q)}{\partial q_i}  = \frac{\partial}{\partial q_i} \frac{q_i}{\sum_j q_j} = \frac{1}{\sum_j q_j} -\frac{q_i}{(\sum_j q_j)^2} = \frac{1 - x_i(q)}{\sum_j q_j} \geq  0,\label{eq:xprime} \end{equation}
and 
\begin{equation} U_i'(q) = \frac{\partial U_i(q)}{\partial q_i} = x_i'(q) - c_i. \label{eq:uprime} \end{equation}

Differentiability of $U_i$ implies that if $\tilde{q}_i$ is a best response to $\tilde{q}_{-i}$, then 
\begin{equation} U_i'(\tilde{q}) \leq 0,\,\,\, \text{ with equality if } \,\,\, \tilde{q}_i > 0. \label{eq:foc} \end{equation} 
In fact, these first-order conditions are sufficient to identify an equilibrium -- that is, any $\tilde{q}$ for which \eqref{eq:foc} holds for all $i$ is an equilibrium. This is because $U_i(q)$ is concave in $q_i$. To see this, note that \eqref{eq:xprime} and \eqref{eq:uprime} imply:
\begin{align*}
\frac{\partial}{\partial q_i}U_i'(q) &= \frac{\partial}{\partial q_i}x_i'(q) = \frac{-x_i'(q)}{\sum_j q_j}  - \frac{1 - x_i(q)}{\left(\sum_j q_j\right)^2} = - \frac{2 x_i'(q)}{\sum_j q_j}\leq 0.\\
\end{align*}

It remains to show that \eqref{eq:xi} and \eqref{eq:foc} are equivalent. First, suppose that \eqref{eq:foc} holds. Then by \eqref{eq:xprime} and \eqref{eq:uprime}, for all $i$ we have
\[ U_i'(\tilde{q}) = x_i'(\tilde{q}) - c_i = (1 - x_i(\tilde{q}))/\sum_j \tilde{q}_j - c_i \leq 0,  \text{ with equality if } q_i > 0. \]
Therefore if $q_i(\tilde{q}) > 0$, we must have  $x_i(\tilde{q}) = 1 - c_i \sum_j \tilde{q}_j$, and $\tilde{q}_i = x_i(\tilde{q}) = 0$ if and only if $c_i \sum_j \tilde{q}_j \geq 1$. In other words, \eqref{eq:xi} holds. Conversely, if \eqref{eq:xi} holds, then substituting into \eqref{eq:xprime} yields
\[x_i'(\tilde{q}) = \frac{1 - x_i(\tilde{q})}{\sum_j \tilde{q}_j} = \frac{1 - \max(1 - c_i \sum_j \tilde{q}_j, 0)}{\sum_j \tilde{q}_j}.\]
Therefore, 

\begin{enumerate}[label=\roman*.]
\item if $\tilde{q}_i > 0$, then $x_i(\tilde{q}) > 0$ and $x'_i(\tilde{q}) = c_i$, so $U_i'(\tilde{q}) = x_i'(\tilde{q}) - c_i = 0$.
\item If $\tilde{q}_i = 0$,  then $x_i'(\tilde{q})  = 1/\sum_j \tilde{q}_j \leq c_i$, so $U_i'(\tilde{q}) = x_i'(\tilde{q}) - c_i \leq 0$.
\end{enumerate}
That is to say, \eqref{eq:foc} holds.\end{proof}

Lemma~\ref{lem:foc} drives the proof of Theorem~\ref{thm:unique2}. We first use Lemma~\ref{lem:foc} to claim that in any equilibrium, miners with lower cost have greater market share.
\begin{corollary}\label{cor:invest}
If $\tilde{q}$ is an equilibrium of the proportional model and $c_i \leq c_j$, then $\tilde{q}_i \geq \tilde{q}_j$. That is, if miner $j$'s costs are no higher than those of miner $i$, then miner $j$ invests at least as much as miner $i$ in equilibrium.
\end{corollary}
\begin{proof}
By Lemma \ref{lem:foc}, $x_i(\tilde{q}) = \max(1 - c_i \sum_j \tilde{q}_j, 0)$, which is weakly decreasing in $c_i$.
%
\end{proof}

\begin{proof}[Proof of Theorem \ref{thm:unique2}]
By Lemma \ref{lem:foc}, we know that $\tilde{q}$ is an equilibrium if and only if $x_i(\tilde{q}) = \max(1 - c_i \sum_j \tilde{q}_j, 0)$ for all $i$. But by definition of $x_i$, we have that 
\[1 = \sum_i x_i(\tilde{q}) = \sum_i \max(1 - c_i \sum_j \tilde{q}_j, 0) = X(1/\sum_j \tilde{q}_j),\]
where $X$ is as defined in \eqref{eq:x}.
Therefore, Lemma \ref{lem:monotone} implies that $\sum_j \tilde{q}_j = 1/c^*$, so that $x_i(\tilde{q}) = \max(1 - c_i /c^*, 0)$ and $\tilde{q}_i = x_i(\tilde{q}) \sum_j \tilde{q} = \frac{1}{c^*}\max(1 - c_i /c^*, 0)$. In other words, there is exactly one equilibrium, and it is as described in Theorem \ref{thm:unique2}. \end{proof}

Of theoretical (if not practical) interest, we note that the condition that $n$ is finite is important for Theorem \ref{thm:unique2} to hold. In particular, in Appendix \ref{app:noeq}, we give an example with countably many miners in which no equilibrium exists.

\section{Disproportional Rewards}\label{sec:disproportional}
In this section, we extend Theorem~\ref{thm:unique2} (or more accurately, something closer to Corollary~\ref{cor:main}) to a model with Economies of Scale (EoS). We call this the \emph{EoS model}. The setup is almost entirely the same: there are still $n \geq 2 $ miners competing for a prize of value $1$. Each miner chooses an investment $q_i$, paying $c_i q_i$ to do so. However, miners now share rewards proportionally to $q_i^\alpha$, for some parameter $\alpha \geq 1$. That is, miner $i$'s reward is $q_i^\alpha /\sum_j q_j^\alpha$. We'll use the same notation as the previous section and denote by $x_i(q) = q_i^\alpha/\sum_j q_j^\alpha$ miner $i$'s market share. The miner's utility is still $x_i(q) - c_i q_i$. We proceed with the main theorem statement of this section:

\begin{theorem} \label{thm:bound} Let $q$ be any equilibrium in the EoS model. Then for all $i,j$ (including $i=j$) such that $q_i, q_j > 0$, $x_i(q) \geq 1 - \frac{1}{\alpha}\frac{c_i}{c_j}$.
\end{theorem}

\begin{remark}
Observe that Theorem~\ref{thm:bound} is almost a strict generalization of Corollary~\ref{cor:main}, except for the assumption that $q_i > 0$. This assumption is necessary, as the conclusion otherwise does not hold (Example~\ref{ex:invest}). 

Observe also that Theorem~\ref{thm:bound} has bite even with perfectly symmetric costs. Theorem~\ref{thm:bound} immediately implies that anyone who participates at all has market share at least $1-1/\alpha$ (by taking $i = j$), so therefore the number of agents who participate in any equilibrium is at most $1+\frac{1}{\alpha-1}$. For example, if $\alpha = 1.05$, then at most $21$ miners will participate. If $\alpha = 1.1$, at most $11$ will. If $\alpha > 2$, there is no pure-strategy equilibrium. 

\end{remark}

\subsection{Proof of Theorem~\ref{thm:bound}}
The proof of Theorem~\ref{thm:bound} is similar to that of Theorem~\ref{thm:unique2}, but does require some new ideas (most notably, Proposition~\ref{prop:important}). We begin with a quick analysis of first-order conditions:

\begin{lemma}\label{lem:dfoc}
Suppose that $q$ is an equilibrium in the EoS model. Then for all $i$, 
\begin{equation}c_i q_i = \alpha \cdot x_i(q) (1 - x_i(q)).  \label{eq:dopt} \end{equation} 
Moreover, for all $i$, $c_i \geq \alpha\cdot q_i^{\alpha-1} \cdot \frac{1-x_i(q)}{\sum_j q^\alpha_j}$.\footnote{Note that this is implied by Equation~\eqref{eq:dopt} when $q_i > 0$, but needs a separate statement when $q_i = 0$.}
\end{lemma}
\begin{proof} Consider the derivative of $x_i(q)$ with respect to $q_i$. We have that:
\[\frac{\partial}{\partial q_i} x_i(q) =  \frac{\partial}{\partial q_i} \frac{q_i^\alpha}{\sum_j q_j^\alpha} = \frac{\alpha q_i^{\alpha-1}}{\sum_j q^\alpha_j} -\frac{\alpha q^{2\alpha-1}_i}{(\sum_j q^\alpha_j)^2} = \alpha x_i(q)/q_i - \alpha x_i(q))^2/q_i = \alpha \frac{x_i(q)\cdot (1-x_i(q))}{q_i} \]

As $U'_i(q) = x'_i(q) - c_i$, we immediately conclude that $U'_i(q) \leq 0$ in equilibrium, with equality if $q_i > 0$. Therefore, $c_i \geq x'_i(q)$, with equality if $q_i > 0$, resulting in Equation~\eqref{eq:dopt}. 
\end{proof}

The following corollary immediately follows (and proves the symmetric case of Theorem~\ref{thm:bound}, where $c_i = c_j$ for all $i,j$). Below, we show that any miner who participates in equilibrium must invest enough to achieve a market share of $1-1/\alpha$. Observe that when $\alpha = 1$, this has no bite, as should be expected given our analysis of the proportional model.

\begin{corollary}\label{cor:dfoc}
Suppose $q$ is an equilibrium in the EoS model. Then for all $i$, $x_i(q) > 0 \Rightarrow x_i(q) \geq 1-1/\alpha$.\footnote{An inquisitive reader might note that this implies that no pure strategy equilibria exists if $\alpha > 2$. Intuitively, when $\alpha$ is large, non-existence may arise for a similar reason that pure-strategy equilibria fail to exist in an all-pay auction: given any investment profile by others, a miner's optimal choice is either to invest just slightly more than the highest competitor, or nothing at all.}
\end{corollary}
\begin{proof}
If $x_i(q) > 0$, then the miner gets non-negative utility from participating, meaning that $x_i(q) \geq c_i q_i = \alpha x_i(q) (1-x_i(q))$. Rearranging yields $1/\alpha \geq 1-x_i(q)$, and then $x_i(q) \geq 1-1/\alpha$.
\end{proof}

We next present an analog of Corollary~\ref{cor:invest}, which states that if there is an equilibrium in which miners $i$ and $j$ both participate, then the miner with lower costs will have a higher market share.

\begin{proposition}\label{prop:important}
Let $q$ be any equilibrium of the EoS model, and let $q_i \geq q_j > 0$. Then $c_i \leq c_j$, and equality holds if and only if $q_i = q_j$.
\end{proposition}

\begin{proof}

Lemma \ref{lem:dfoc} states that $c_i q_i = \alpha x_i(q) (1 - x_i(q))$, and analogously for $j$. If $q_i, q_j > 0$, we can note that $q_i = x_i(q)^{1/\alpha} \left(\sum_\ell q_\ell^\alpha \right)^{1/\alpha}$ and divide by $x_i(q)^{1/\alpha}$ to get
\[ c_i \left(\sum_\ell q_\ell^{\alpha}\right)^{1/\alpha} = \alpha x_i(q)^{1 - 1/\alpha} (1 - x_i(q)) = \alpha f(x_i(q)),\]
where we define $f(x) = x^{1 - 1/\alpha}(1 - x)$.

We know by Corollary \ref{cor:dfoc} that $x_i(q)$ and $x_j(q)$ are both at least $1 - 1/\alpha$. Because $x_i(q) \geq x_j(q) \Leftrightarrow q_i \geq q_j$, to prove the Proposition, it suffices to show that $f$ is decreasing on $(1 - 1/\alpha, 1)$. But
\begin{align*}
f'(x) & = (1-1/\alpha)x^{-1/\alpha}(1-x) -x^{1-1/\alpha} =(1 - 1/\alpha) x^{-1/\alpha} - (2 - 1/\alpha) x^{1 - 1/\alpha} \\
&  = \frac{2 \alpha - 1}{\alpha}  x^{-1/\alpha}\left( \frac{\alpha - 1}{2 \alpha - 1} - x\right).
\end{align*}
Thus, $f'(x) < 0$ for all $x > \frac{\alpha - 1}{2 \alpha - 1}$, and in particular for all $x > \frac{\alpha - 1}{\alpha}$, since $\frac{\alpha - 1}{\alpha} \geq \frac{\alpha - 1}{2\alpha - 1}$.\end{proof}

Proposition~\ref{prop:important} is the main workhorse in the proof of Theorem~\ref{thm:bound}, and the remainder proof of the proof now follows quite easily.

\begin{proof}[Proof of Theorem \ref{thm:bound}]
Let miner $i$ and miner $j$ both participate in $q$, and let $c_i \leq c_j$. Therefore, by Proposition~\ref{prop:important}, $q_j \leq q_i$. We have from \eqref{eq:dopt} that 
\[ \frac{c_i}{c_j} \frac{q_i}{q_j} = \frac{x_i (1 - x_i)}{x_j (1 - x_j)}.\]
Solving for $1 - x_i$ and using the fact that $x_i/x_j = q_i^\alpha/q_j^\alpha$, we get:
\begin{align*}
(1-x_i) &= \frac{c_i}{c_j}\cdot \frac{x_j}{x_i} \cdot \frac{q_i}{q_j} \cdot (1-x_j) = \frac{c_i}{c_j}\cdot \left(\frac{q_j}{q_i}\right)^{\alpha-1} \cdot (1-x_j) \leq \frac{1}{\alpha} \cdot \frac{c_i} {c_j}\\
\Rightarrow x_i &\geq 1 - \frac{1}{\alpha}\frac{c_i}{c_j}
\end{align*}
The final step follows because $1 - x_j \leq \frac{1}{\alpha}$ (Corollary~\ref{cor:dfoc}) and $q_j/q_i \leq 1$ (hypothesis + Proposition~\ref{prop:important}). This takes care of the case where $c_i \leq c_j$. When $c_i \geq c_j$, the theorem immediately follows from Corollary~\ref{cor:dfoc}. 
\end{proof}


As an aside for the interested reader, we note that the condition in Proposition \ref{prop:important} that both miners choose to invest is necessary. As the following example shows, when $\alpha > 1$, it is possible that one miner invests in equipment, while a lower-cost miner abstains. Intuitively, this occurs because in the presence of economies of scale, an inefficient incumbent can deter a more efficient potential entrant. 

\begin{example}\label{ex:invest}
Fix an integer $m\geq 2$, and consider the case with $n > m$ miners, and $\alpha =m/(m-1)$, where $c_1 = (1-1/m)^{1-1/m} < 1 = c_2 = c_3 = \cdots = c_n$. Then for any $M \subseteq \{2, 3, \ldots, n\}$ with $|M| = m$, there exists an equilibrium where $q_{i} =1/m$ for $i \in M$, and $q_i = 0$ for $i \not \in M$ (see Appendix~\ref{app:numerical} for a proof). When $m = \alpha = 2$, these are in fact the only equilibria. Observe that in all of these equilibria, $q_1 = 0$, despite the fact that $c_1$ is the lowest cost.
\end{example}

\section{Conclusions and Discussion}\label{sec:conclusions}

In this paper, we argue that in the presence of minor cost asymmetries and economies of scale, Bitcoin's proof-of-work protocol incentivizes extreme concentration of mining power. The observed dominance of large pools, which are themselves dominated by a few large mining entities, is not a temporary aberration, but rather is likely to persist for the foreseeable future. This conclusion has two broad implications for future cryptocurrency research. First, people interested in understanding and making predictions about Bitcoin should take into consideration that the vision of a large competitive market among miners is unlikely to be fulfilled. For example, existing research often assumes that miners are ``small"  \cite{CarlstenKNW16,Huberman-Leshno-Moallemi_2017}. In light of our findings, it seems especially important to understand which conclusions are sensitive to this assumption.

Second, our findings motivate research into alternate cryptocurrencies in which cost asymmetries and economies of scale are less pervasive. Much of the cost of Bitcoin mining is driven by costs associated with powering and cooling hardware.\footnote{For instance, one can purchase a state-of-the-art AntMiner S9 for roughly 2000 USD. The estimated cost to power it for a year in an average US home is roughly the same amount. Larger miners achieve roughly the same ratio: their electricity is considerably cheaper and they also receive bulk discounts for initial purchase (\url{https://www.buybitcoinworldwide.com/mining/profitability/}).} Because electricity is difficult to store and transport, its price varies widely with geography, implying that potential miners in different locations face very different costs. Alternative solution concepts like proof-of-stake~\cite{DaianPS17, gilad2017algorand, KiayiasRDO17}, where miners are selected in proportion to their wealth (in the currency itself) rather than computational power, might address this issue: while the cost of purchasing cryptocurrency may exhibit minor asymmetries across investors, it is likely nowhere near the level of asymmetry observed in electricity prices.

Our work also highlights the importance of minimizing economies of scale. Ideally, a miner's expected return would be concave in their share of mining power. This seems difficult to achieve in practice, as miners can always divide their hardware among several false identities. At the very least, however, one would like to ensure that large Bitcoin miners do not earn rewards {\em exceeding} their share of total mining power. Currently, this might occur due to orphaned blocks or strategic manipulations of the Bitcoin protocol. Below, we discuss each of these factors, and how they might be mitigated by alternative cryptocurrency designs.


Due to network latency, it is possible that two miners find solutions to the cryptopuzzle at ``the same'' time (meaning close enough that initially, neither is aware of the other). These solutions and their associated block of transactions cannot both be added to the ledger (this is part of Bitcoin's ``longest-chain protocol''), so the next miner to find a block determines which block is added (earning a reward for its miner) and which is ``orphaned'' (earns no reward). Miners are asked to tie-break in favor of the block they heard about first, implying that larger miners are less likely to see their blocks orphaned (because they are more likely to mine the subsequent block, and heard about their own block first). Therefore, our work further motivates the design of protocols such as GHOST~\cite{SompolinskyZ13}, which aims to greatly reduce the number of orphaned blocks, or Algorand~\cite{ChenM17, gilad2017algorand}, which removes the concept of an orphaned block altogether. 

While the Bitcoin protocol specifies how miners are {\em supposed} to behave, it cannot enforce this behavior. Existing work \cite{EyalS14, SapirshteinSZ16, KiayiasKKT16, CarlstenKNW16} observes that large miners may benefit from ``selfish mining" or waiting to announce a block. For instance, these strategies may allow a miner with (say) 10\% of the total computational power to reap (say) 11\% of the total rewards in most existing cryptocurrencies (these are not exact numbers). While most researchers recognize such deviations as problematic, the minor gains are often cited as a mitigating factor (after all, the remaining 90\% of miners still enjoy 89\% of the total rewards). Our work demonstrates that the possibility of such gains may encourage mining power concentration, and therefore supports the importance of designing incentive compatible protocols.

Our model assumes that the value of the block reward is exogenous to the behavior of Bitcoin miners. While the frequency of block creation and the number of bitcoins awarded to the block creator are fixed by the Bitcoin protocol, the value of these bitcoins in, say, US Dollars is determined by market forces. It is plausible that the price of Bitcoin is influenced by its perceived stability, in which case Bitcoin's value might drop as mining power became more concentrated, and rise as (aggregate) mining power increased. Because we believe that most fluctuations in Bitcoin's price are driven by speculation rather than assessments of its fundamentals, we leave a formal treatment of these effects to future work.

Another abstraction in this paper is to treat the acquisition of mining power as a one-time investment, when in reality new hardware can be purchased at any time. Because we focus on pure-strategy equilibria, so long as the value of Bitcoin (and cost of hardware) is stable over time, miners will have no incentive to wait to acquire hardware. In reality, of course, fluctuating Bitcoin prices may provide miners with incentives to acquire new hardware, or to stop powering existing hardware. While a full analysis of these incentives is beyond the scope of this work, we note that adding dynamics to the game doesn't change the fact that miners are investing resources to compete over a fixed prize, nor does it change our conclusion that seemingly small cost asymmetries or advantages to scale provide incentives for mining power to become highly concentrated.

Finally, recall that we choose to isolate two key factors (cost asymmetry and economies of scale) and demonstrate that these factors \emph{alone} transparently imply centralization of mining power. There are numerous nuances to Bitcoin mining (in addition to those described immediately above) which are certainly interesting to investigate in future work, but explicitly modeling these nuances would not change the underlying thesis of this paper: proof of work protocols result in high concentration of mining power in the presence of even minor cost asymmetries or economies of scale.

\newpage
\bibliographystyle{alpha}
\bibliography{MasterBib}

\newpage
\appendix

\section{Proof of Example~3}\label{app:numerical} 
Before getting into our proof, we'll need the following facts about equilibria in the EoS model. In the proportional model, utilities are concave in investment, so first-order conditions suffice to verify an equilibrium. In the EoS model, utilities are not concave in investment. In principle, this means one might need to appeal to global optimality conditions, but fortunately there is still some similar structure, which will be crucial in order to tractably analyze equilibria. Lemma~\ref{lem:convexconcave} provides ``concave-like'' structure on the utilities, and is used to prove Corollary~\ref{cor:unique}, which is our key structural insight for examples in the EoS model.

\begin{lemma}\label{lem:convexconcave}
Utilities in the EoS model are initially convex, and then concave. Specifically, $U_i(q)$ is strictly convex in $q_i$ when $x_i(q) \in [0,\frac{\alpha-1}{2\alpha})$ and strictly concave when $x_i(q) \in (\frac{\alpha-1}{2\alpha},1)$.
\end{lemma}
\begin{proof}
As in Lemma~\ref{lem:dfoc}, recall that $\frac{\partial}{\partial q_i} x_i(q) = \alpha x_i(q) (1-x_i(q))/q_i$ and $\frac{\partial}{\partial q_i} U_i(q) = \alpha\cdot x_i(q) \cdot (1-x_i(q)) /q_i - c_i$. So:

\begin{align*}
\frac{\partial}{\partial q_i} \left(\frac{\partial}{\partial q_i}U_i(q)\right) &= \alpha \cdot \frac{(1-2x_i(q))\cdot \frac{\partial x_i(q)}{\partial q_i} }{q_i} - \alpha \cdot \frac{x_i(q) (1-x_i(q))}{q_i^2}\\
&=\alpha \cdot \frac{ (1-2x_i(q))\cdot \frac{\partial x_i(q)}{\partial q_i}}{q_i} - \alpha \cdot \frac{\frac{q_i}{\alpha}\cdot \frac{\partial x_i(q)}{\partial q_i}}{q_i^2}\\
&= \frac{\frac{\partial x_i(q)}{\partial q_i}}{q_i} \cdot (\alpha - 2x_i(q) - 1).
\end{align*}
Observe that $\frac{\partial x_i(q)}{\partial q_i} > 0$, and $q_i > 0$, and $\alpha - 2\alpha x_i(q) - 1$ is initially positive but strictly decreasing in $q_i$, hitting $0$ when $x_i(q) = \frac{\alpha-1}{2\alpha}$. 
\end{proof}

\begin{corollary}\label{cor:unique} For all $q_{-i}$, the best response of miner $i$ is either to invest the value $q^*_i$ such that $x_i(q^*_i;q_{-i}) \geq \frac{\alpha-1}{2\alpha}$ and $\frac{\partial}{\partial q_i} U_i(q^*_i; q_{-i}) = 0$, or to invest nothing if such a $q_i^*$ does not exist. 
\end{corollary}

Now we proceed to analyze the example. Recall that we wish to show that when $\alpha = m/(m-1)$, $c_1$ is sufficiently large and $c_i = 1$ for all $i \geq 2$, that it is an equilibrium for any $m$ miners with index $\geq 2$ to each invest $1/m$, and all other miners to invest $0$.

First, consider the reward of any miner who is investing at the proposed strategy profile: it is exactly $1/m- 1/m = 0$. Moreover, $i$'s market share at the proposed strategy profile is $1/m > \frac{\alpha-1}{2\alpha} = \frac{1}{2m}$. By Corollary~\ref{cor:unique}, it therefore suffices to show that first-order conditions are satisfied to prove that miner $i$ is best responding. By Lemma~\ref{lem:dfoc}, first-order conditions are satisfied, as $c_i q_i = 1/m = \alpha \cdot 1/m \cdot (m-1)/m = \alpha \cdot x_i(q) (1-x_i(q))$. 

Finally, we just need to make sure that all miners who choose not to participate are best responding as well. Note that because there are $m$ other miners each investing $1/m$, $x_1(q) = q_1^\alpha/(q_1^\alpha + m (1/m)^\alpha) = q_1^\alpha/(q_1^\alpha + m^{-1/(m-1)})$. Therefore:

\begin{align*}
x_1(q) &=\frac{q_1^\alpha}{q_1^\alpha + m^{-1/(m-1)}}\\
&\Rightarrow q_1^\alpha x_1(q) + x_1(q) m^{-1/(m-1)} = q^\alpha_1\\
&\Rightarrow q_1^\alpha (1-x_1(q)) = x_1(q) m^{-1/(m-1)}\\
\Rightarrow q_1 = \left(x_1(q) m^{-1/(m-1)}/(1-x_1(q))\right)^{1/\alpha}
\end{align*}

Now, we want to make use of Lemma~\ref{lem:dfoc}, which states that if miner $1$ is best responding, then $c_1 q_1 = \alpha x_1(q) (1-x_1(q))$. Specifically, we wish to find the value of $c_i$ such that miner $1$'s best response is to target a market share of $1/m$. Continuing from above, we get:

\begin{align*}
&x_1(q) = 1/m, \quad 1-x_1(q) = 1-1/m, \quad \alpha = 1+1/(m-1), \quad q_1 = \left(x_1(q) \cdot m^{-1/(m-1)}/(1-x_1(q))\right)^{1/\alpha},\\& \quad c_1 q_1 = \alpha x_1(q) (1-x_1(q))\\
&\Rightarrow q_1 = \left(m^{-1/(m-1)}/(m-1)\right)^{(m-1)/m} =m^{-1/m}\cdot (m-1)^{-(m-1)/m}.\\
&\Rightarrow c_1 m^{-1/m}\cdot (m-1)^{-(m-1)/m} = \frac{m}{m-1} \cdot \frac{1}{m} \cdot \frac{m-1}{m} = 1/m.\\
&\Rightarrow c_1 = m^{-1+1/m} \cdot (m-1)^{1-1/m} = (1-1/m)^{1-1/m}.\\
\end{align*}

This implies that when $c_1 = (1-1/m)^{1-1/m}$, the best response of miner $1$ is either to target a market share of $1/m$, or invest $0$. In fact, both are best responses and yield $0$ utility, as we solved for $c_1 q_1 = 1/m$. This allows us to conclude that if $c_1 = (1-1/m)^{1-1/m}$, it is a best response for miner $1$ to not invest. This further allows us to conclude that when $c_1 > (1-1/m)^{1-1/m}$, it is a best response for miner $1$ to not invest (as all strategies aside from not investing become strictly less attractive). Therefore, miner $1$'s best response is to not invest whenever $c_1 \geq (1-1/m)^{1-1/m}$.

To see that these are the only possible equilibria when $\alpha = 2$, observe by Corollary~\ref{cor:dfoc} that in any equilibrium there must be exactly two miners who participate and they must have market share exactly $1/2$. Assume for contradiction that one of them is miner $1$ (and the other is some $i > 1$). Lemma~\ref{lem:dfoc} then implies that $c_1q_1 = 1/2 = c_i q_i$. But if they each have market share $1/2$, we must have $q_1 = q_i$. As $c_1 < c_i$, there is no equilibrium of this form. 

\section{No Equilibrium when $n = \infty$}\label{app:noeq}
In this section we show an example (based on Example~\ref{ex:exp}) with countably many miners where no equilibrium exists, despite the fact that an equilibrium exists for every finite prefix, \emph{and} the sequence of equilibria converge as $n \rightarrow \infty$. 

\begin{example}\label{ex:noeq} For the rest of this section, let $c_i = 1-2^{-i-k}$, for any $k > 0$. 
\end{example}

\begin{proposition} In Example~\ref{ex:noeq}, there is no equilibrium where finitely many miners participate.

\end{proposition}

\begin{proof}
Assume for contradiction that an equilibrium $q$ exists with finitely many miners, and let $n$ denote the least efficient miner to participate. Consider instead the finite instance that contains only miners $1,\ldots, n+1$. Then clearly $q^{(n+1)}$ ($q$ restricted to its first $n+1$ coordinates) is an equilibrium for this game. Specifically, as far as the first $n+1$ miners can tell, $q$ is identical with or without miners $> n+1$, because they don't participate anyway. So if each of them are best responding in $q$, they are certainly best responding in $q^{(n+1)}$. 

By Theorem~\ref{thm:unique2}, there is a unique such equilibrium. Observe that $X(1) = \sum_i (1-c_i) < 1$, from which it follows that $c^* > 1$ (because $X(c^*) =1$ and $X$ is non-decreasing). Therefore, all miners choose to participate as $c_i < c^*$. However, by definition of $n$, miner $n+1$ does \emph{not} participate in $q$ (and therefore also does not participate in $q^{(n+1)}$), so $q^{(n+1)}$ is not an equilibrium for the game restricted to the first $n+1$ miners, and therefore $q$ cannot be an equilibrium. 
\end{proof}

\begin{observation}\label{obs:noeq} When Example~\ref{ex:noeq} is restricted to $n < \infty$ miners, every miner chooses to participate and miner $i$'s market share is $\frac{1-2^{-k}+2^{-k-n}+(n-1)2^{-i-k}}{n-2^{-k}+2^{-k-n}}$. Observe that for all $i$, as $n \rightarrow \infty$, the market share approaches $2^{-i-k}$ from above.

\end{observation}

Observation~\ref{obs:noeq} suggests that there should be an equilibrium of Example~\ref{ex:noeq} where miner $i$'s market share is $2^{-i-k}$. Something of course seems off about this, because the market shares would then sum to $2^{-k} < 1$ (as $k > 0$). Indeed, there is no such equilibrium, and in fact no equilibrium at all.

\begin{proposition} In Example~\ref{ex:noeq}, there is no equilibrium where infinitely many miners participate. 
\end{proposition}
\begin{proof}
Assume for contradiction that some equilibrium exists, and let $T=\sum_i q_i$ denote the total investment in mining. Note that $T < \infty$ as the total reward is $1$, and the minimum possible cost for investment $T$ is $T \cdot (1-2^{-1-k})$. So we must have $T \leq 1/(1-2^{-1-k})$. 

Now consider any miner $i$, who invests $q_i$. Then their reward for choosing to invest $q'_i$ instead is exactly $\frac{q'_i}{T-q_i+q'_i}- c_iq'_i$, whose derivative is $\frac{1}{T-q_i + q'_i} - \frac{q'_i}{(T-q_i + q'_i)^2} - c_i = \frac{T-q_i - c_i(T-q_i + q'_i)^2}{(T-q_i+q'_i)^2}$. In fact, miner $i$'s best response is to invest $\max\{0,\sqrt{(T-q_i)/c_i} - (T-q_i)\}$. So if miner $i$ is investing at all, we must have $q_i = \sqrt{(T-q_i)/c_i}-(T-q_i)$. Therefore, we get the following implications (where $S$ denotes the set of miners who participate in the assumed equilibrium):
\begin{align*}
q_i &= \sqrt{(T-q_i)/c_i} - T + q_i\\
&\Rightarrow T^2 = T/c_i - q_i/c_i\\
&\Rightarrow q_i = c_iT^2 - T\\
&\Rightarrow \sum_{i \in S} q_i = \sum_{i \in S}\left( c_i T^2 - T\right)\\
&\Rightarrow 1 = \sum_{i \in S} \left(c_i T - 1\right).\\
\end{align*}
Note that this equation simply doesn't hold for any $T$. If $T < 1$, then the RHS diverges to $-\infty$ (recall that $S$ is countably infinite). If $T > 1$, then the RHS diverges to $+\infty$. If $T = 1$, then the RHS is upper bounded by $2^{-k}$ (in the case that $S = \mathbb{N}$), which is strictly less than one.

\end{proof}

The above propositions combine to yield Theorem~\ref{thm:noeq}.
\begin{theorem}\label{thm:noeq}In Example~\ref{ex:noeq}, there exists a unique equilibrium for every prefix of $n < \infty$ miners, and miner $i$'s investment in these equilibria converges from above to $2^{-k-i}$. Yet, there exists no equilibrium for Example~\ref{ex:noeq} (or in fact any infinite subset of miners).

\end{theorem}

\begin{comment}
\section{Example with Vanishing Fraction of Market Participants}\label{app:examples}
\begin{example}
Let costs be evenly spaced on $[a,1]$ -- that is, $c_i = a + \frac{1-a}{n} i$. We claim that the number of miners entering is approximately $\sqrt{ \frac{2an}{1 - a}}$, and the market share of the lowest-cost miner is $\sqrt{\frac{2(1-a)}{an}} + o(n^{-1/2})$. In other words, as $n$ grows, the number of entrants tends to infinity, the market share of each entrant tends to zero, and for any $\epsilon > 0$, no miner with cost at least $a (1+\epsilon)$  (that is, with costs at least $(1 + \epsilon)$ times the lowest mining cost) will enter. 

Note also that the fraction of market participants vanishes.
\end{example}

Write $c^*$ as $c^* = a + \frac{1-a}{n} (k + \delta)$ for some $\delta \in (0,1]$, so that the first $k$ miners all have positive market share. Then by definition, $X(c^*) = \sum_{i = 1}^k (1 - c_i/c^*) = k - \sum_{i = 1}^k c_i/c^* = 1$, from which it follows that 
\[c^*k- c^* = \sum_{i = 1}^k c_i = k a + \frac{1-a}{n} \frac{k(k+1)}{2}. \]
Moving the terms on the right side to the left, and multiplying by $2n/(1-a)$, we get
\[k \frac{2n}{1-a} (c^* - a) - \frac{2n}{1 - a} c^*  - k(k+1) = 0.\]
Substituting our expression for $c^*$, we get that
\[2k(k+\delta)  - \frac{2n}{1 - a} (a+\frac{1-a}{n} (k + \delta)) - k(k+1)= 0,\]
which after rearrangement yields
\[k^2 - 3k +2\delta (k-1) - \frac{2an}{1 - a} = 0.\]
The left side is increasing in $\delta$. Because $0 < \delta \leq 1$, we have that
\[(k - 3/2)^2 - 9/4 = k^2 - 3k  < \frac{2an}{1 - a} \leq k^2 - k - 2 < (k - 1/2)^2 \]
From this, we conclude that 
\[\frac{1}{2} + \sqrt{\frac{2an}{1 - a}} < k < \frac{3}{2} + \sqrt{ \frac{9}{4} + \frac{2an}{1 - a}} < 3 + \sqrt{\frac{2an}{1 - a}}.\]
where the final inequality follows because $\sqrt{x+y} < \sqrt{x} + \sqrt{y}$.

\end{document}